\documentclass[preprint]{elsarticle}
\makeatletter
\def\ps@pprintTitle{%
  \let\@oddhead\@empty
  \let\@evenhead\@empty
  \let\@oddfoot\@empty
  \let\@evenfoot\@oddfoot
}
\makeatother

\usepackage{amsmath,amsthm,amscd,bm}

\usepackage{tikz}
\usetikzlibrary{matrix,decorations.pathreplacing, calc, positioning}

\usepackage{eucal}
\usepackage[charter,cal=cmcal]{mathdesign}
\usepackage[utf8]{inputenc}
\usepackage[T1]{fontenc}
\usepackage[colorlinks]{hyperref}
\usepackage{url}
\usepackage{microtype}
\usepackage{float}
\usepackage{todonotes}
\usepackage{aurical,bbm}

\newcommand{\al}{\Sigma}
\newcommand{\pow}{\ell}

\newcommand{\N}{\ensuremath{\mathbb{N}}\xspace}
\newcommand{\F}{\ensuremath{\mathbb{F}}\xspace}

\newcommand{\set}[1]{\{#1\}}

\newcommand{\myvec}[1]{\ensuremath{\boldsymbol{#1}}\xspace}

\renewcommand{\P}{\myvec{P}}

\newcommand{\ccc}{\myvec{c}}

\usepackage{aurical,bbm}
\newcommand{\coeff}{\alpha}
\newcommand{\pol}{\pi}

\renewcommand{\P}{\myvec{P}}

\newcommand{\Z}{\mathbb{Z}}

\newcommand{\LP}{\mathbb{L}}
\newcommand{\LS}{\mathbb{S}}
\newcommand{\pid}{\mathbb{P}}
\newcommand{\field}{\mathbb{F}}

\newcommand{\zetaki}{\Z/k_i\Z}
\newcommand{\zetapk}{\Z/p^k\Z}
\newcommand{\zetapku}{\Z/p^{k_1}\Z}

\newcommand{\zetapkn}{\Z/p^{k_n}\Z}

\newcommand{\LPKuno}{\ensuremath{\Z/p^{k_1}\Z\left[X,X^{-1}\right]}\xspace}

\newcommand{\emb}{\psi}
\newcommand{\Emb}{\Psi}
\newcommand{\locrule}{\ensuremath{f}}
\newcommand{\glorule}{\ensuremath{{\cal F}}}
\usepackage{xspace, comment}

\newcommand{\ie}{i.e.\@\xspace}

\newcommand{\az}{S^{\mathbb{Z}}}
\newcommand{\zd}{\mathbb{Z}^D}

\newcommand{\gzd}{{\al}^{\Z}}

\usepackage{pst-node}

\usepackage{tikz-cd} 
\usetikzlibrary{arrows}

\usepackage{clrscode3e}
\usepackage{geometry}
\usepackage{varwidth}
\newcommand{\mmod}{\,\mathrm{mod}\, }

\newtheorem{theorem}{Theorem}
\newdefinition{remark}[theorem]{Remark}
\newdefinition{definition}[theorem]{Definition}
\newdefinition{example}[theorem]{Example}
\newtheorem{lemma}[theorem]{Lemma}

\newtheorem{corollary}[theorem]{Corollary}

\newcounter{exnum}
                        {\hfill\hbox{$\quad\Box$}\end{trivlist}}

\begin{document}
\title{An Easily Checkable Algebraic Characterization of Positive Expansivity  for Additive Cellular Automata over a Finite Abelian Group
}
\author[unimib]{Alberto Dennunzio}
\ead{alberto.dennunzio@unimib.it}

\author[UCA]{Enrico Formenti}
\ead{enrico.formenti@univ-cotedazur.fr}


\author[unibo]{Luciano Margara}
\ead{luciano.margara@unibo.it}

\address[unimib]{Dipartimento di Informatica, Sistemistica e Comunicazione,
  Università degli Studi di Milano-Bicocca,
  Viale Sarca 336/14, 20126 Milano, Italy}
  
\address[UCA]{Universit\'e C\^ote d'Azur, CNRS, I3S, France}


\address[unibo]{Department of Computer Science and Engineering, University of Bologna, Cesena Campus, Via dell'Universit\`a 50, Cesena, Italy}

\begin{abstract}
We provide an easily checkable algebraic characterization of positive expansivity  for Additive Cellular Automata over a finite abelian group. First of all, an easily checkable characterization of positive expansivity is provided for the non trivial subclass of Linear Cellular Automata over the alphabet $(\Z/m\Z)^n$. Then, we show how it can be exploited to decide positive expansivity for the whole class of Additive Cellular Automata over a finite abelian group. 
\end{abstract}

\begin{keyword}
 Cellular Automata  \sep Additive Cellular Automata \sep Chaos \sep Positive Expansivity
\end{keyword}

\maketitle

\section{Introduction}
We provide an easily checkable algebraic characterization of positive expansivity  for \emph{Additive Cellular Automata}  over a finite abelian group. 
First of all, an easily checkable algebraic characterization of positive expansivity is provided for the non trivial subclass of Linear Cellular Automata over the alphabet $(\Z/m\Z)^n$, where $m$ is any natural greater than 1. Then, we show how it can be exploited to decide positive expansivity for the whole class of Additive Cellular Automata over a finite abelian group. 

The main and more difficult part of this work consists of the proof of an algebraic characterization of  
positive expansivity for Linear Cellular Automata over $(\Z/p\Z)^n$, where $p$ is any prime number. To reach that result
\begin{enumerate}
\item first of all, we provide an easily checkable algebraic characterization of positively expansive Linear Cellular Automata over $(\Z/p\Z)^n$ with associated matrix that is (the traspose of one) in a rational canonical form consisting of only one block; this is the heart of our work; 
\item then, we prove that such an algebraic characterization turns out to hold also for positively expansive Linear Cellular Automata over $(\Z/p\Z)^n$ with associated matrix that is in a rational canonical form possibly consisting of 
more than one block; 
\item finally, we prove that such an algebraic characterization turns out to hold also for all positively expansive Linear Cellular Automata over $(\Z/p\Z)^n$.
 \end{enumerate}
Afterwards, the algebraic characterization of positive expansivity is extended from Linear Cellular Automata over $(\Z/p\Z)^n$ to Linear Cellular Automata over $(\Z/p^k\Z)^n$, where $k$ is any non zero natural. Finally, we show how such a characterization can be exploited to decide positive expansivity for Linear Cellular Automata over $(\Z/m\Z)^n$ and, at the end, for whole class of Additive Cellular Automata over a finite abelian group.

\section{Basic notions}
\label{sec:basic}
%

Let $\mathbb{K}$ be any commutative ring and let $A\in\mathbb{K}^{n\times n}$ be an $n\times n$-matrix over $\mathbb{K}$. We denote by $A^T$ the transpose matrix of $A$ and by 
$\chi_{A}$ the characteristic polynomial $\det\left(
tI_{n}-A\right)  \in\mathbb{K}\left[  t\right] $ of $A$,  where $I_n$ always stands for the 
$n\times n$ identity matrix (over whatever ring we are considering).  Furthermore,  $\mathbb{K}[X,X^{-1}]$ and $\mathbb{K}[[X,X^{-1}]]$ denote the set of Laurent polynomials and series, respectively, with coefficients in $\mathbb{K}$. In particular, whenever $\mathbb{K}=\Z/m\Z$ for some natural $m> 1$, we will write $\LP_m$ and $\LS_m$ instead of $\Z/m\Z[X,X^{-1}]$ and $\Z/m\Z[[X,X^{-1}]]$, respectively. 

Let $\mathbb{K}=\Z/m\Z$ for some natural $m> 1$ and let $q$ be a natural with $q<m$. If $P$ is  any polynomial from $\mathbb{K}[t]$ (resp., a Laurent polynomial from $\LP_m$) (resp., a matrix from $(\LP_m)^{n\times n}$), $P\mmod q$ denotes the polynomial (resp., the Laurent polynomial) (resp., the matrix) obtained by $P$ by taking all its coefficients modulo $q$.




\smallskip

Let $\al$ be a finite set (also called \emph{alphabet}). A \emph{CA configuration}  (or, briefly, a \emph{configuration}) is any function from $\Z$ to $\al$.  Given a configuration $c\in\gzd$ and any 
integer $i\in\Z$, the value of $c$ in position $i$ is denoted by $c_i$. 
The set $\gzd$, called \emph{configuration space}, is as usual 
equipped with the standard Tychonoff distance $d$.
Whenever the term \emph{linear} is involved the alphabet $\al$ is $\mathbb{K}^n$, where $\mathbb{K}=\Z/m\Z$ for some natural $m> 1$.  
Clearly, in that case both $\mathbb{K}^n$ and $(\mathbb{K}^n)^\Z$ become $\mathbb{K}$-modules in the obvious (\ie, entrywise) way.  On the other hand, whenever the term \emph{additive} is involved the alphabet $\al$ is a finite abelian group $G$ and the configuration space turns $G^\Z$ turns out to be an abelian group, too, where the group operation of $G^\Z$ is the componentwise extension of the group operation of $G$, both of them will be denoted by $+$.

\smallskip

A \emph{one-dimensional CA} (or, briefly, a \emph{CA}) over $\al$ 
is a pair $(\gzd,\glorule)$, where $\glorule\colon\gzd\to\gzd$
is the uniformly continuous transformation (called \emph{global rule}) defined as 
$
\forall c\in\gzd, \forall i\in\Z, \glorule(c)_i=\locrule(c_{i-r}, \ldots, c_{i+r}),%
$
for some fixed natural number $r\in\N$ (called \emph{radius}) and some fixed function $\locrule\colon {\al}^{2r+1}\to \al$ (called  \emph{local rule} of radius $r$).   
In the sequel, when no misunderstanding is possible, we will sometimes identify any CA with its global rule. 

\smallskip

We recall that a CA $(\gzd,\glorule)$ is \emph{positively expansive} if for some constant $\varepsilon>0$ it holds that for any pair of configurations $c,c'\in$ there exists a natural number $\pow$ such that $d(\glorule^{\pow}(c),\glorule^{\pow}(c'))\geq\varepsilon$. We stress that CA positive expansivity is a condition of strong \emph{chaos}.  Indeed, on a hand, positive expansivity is a stronger condition than \emph{sensitive dependence on the initial conditions}, the latter being the essence of the chaos notion. On the other hand, any positively expansive CA is also \emph{topologically transitive} and, at the same time, it has \emph{dense periodic orbits}. Therefore, any positively expansive CA is chaotic according to the Devaney definition of chaos (see~\cite{devaney89}, for the definitions of chaos, sensitive dependence on the initial conditions, topological transitivity, and denseness of dense periodic orbits). Finally, we recall that if a CA $\glorule$ is positively expansive then $\glorule$ is surjective.
\paragraph{Linear and Additive CA} 
Let $\mathbb{K}=\Z/m\Z$ for some natural $m>1$ and let $n\in\N$ with $n\geq 1$. Let $G$ be a finite abelian group.

\smallskip

A local rule $\locrule\colon (\mathbb{K}^n)^{2r+1}\to \mathbb{K}^n$ of radius $r$ is said to be \emph{linear}
if it is defined by  $2r+1$ matrices  
$A_{-r},\ldots, A_r\in \mathbb{K}^{n\times n}$
as follows: 
$
\forall (x_{-r}, \ldots,x_r)\in(\mathbb{K}^n)^{2r+1}, \locrule(x_{-r},\ldots,x_r) = 
\sum_{i=-r}^r A_i\cdot x_i
\enspace.
$ 
%
%
A one-dimensional \emph{linear CA (LCA)} over $\mathbb{K}^n$ is a CA \glorule\xspace
based on a linear local rule. 
The Laurent polynomial (or matrix)
\[
A=\sum\limits_{i=-r}^{r} A_i X^{-i}\in\mathbb{K}^{n\times n}[X,X^{-1}]\cong (\LP_m)^{n\times n}
\]
is said to be the 
\emph{the matrix} \emph{associated with \glorule}. 

\smallskip 


We now recall the notion of Additive CA, a wider class than LCA. An \emph{Additive CA} over $G$ is a CA $(G^\Z,\glorule)$ where the global rule $\glorule: G^\Z\to G^\Z$ is an endomorphism of $G^\Z$. 
We stress that the local rule $\locrule: G^{2r+1}\to G$ of an Additive CA of radius $r$ over a finite abelian group $G$ can be written as
$\forall (x_{-r}, \ldots, x_{r})\in G^{2r+1},  \locrule(x_{-r},\dots,x_r)=\sum_{i=-r}^{r} \locrule_i(x_i)$, 
where the functions $\locrule_i$ are endomorphisms of $G$. Moreover, as a consequence of the application of the fundamental theorem of finite abelian groups to Additive CA (see~\cite{DennunzioFGM21INS}, for details), without loss of generality we can assume that 
$G=\zetapku\times\ldots\times \zetapkn$ for some naturals $k_1, \ldots, k_n$ with $k_1\geq k_2\geq \ldots\geq k_n$.

Let $\hat{G}=(\zetapku)^n$ and let $\emb: G\to \hat{G}$ be the map defined as 
$
\forall h\in G, \forall i=1,\dots,n, \emb(h)^i=h^i\, p^{k_1-k_i}\enspace,
$
where, for a sake of clarity, we stress that $h^i$ denotes the $i$-th component of $h$, while $p^{k_1-k_i}$ is just the $(k_1-k_i)$-th power of $p$. Let $\Emb: G^{\Z}\to \hat{G}^{\Z}$ be the componentwise extension of $\emb$, \ie, the function defined as 
$
\forall \ccc\in G^{\Z},  \forall j\in\Z,  \Emb(\ccc)_j=\emb(\ccc_j).
$
The function $\Emb$ turns out to be continuous and injective. 

\smallskip
We recall that for any Additive CA over $G$ an LCA over $(\zetapku)^n$ associated with it can be defined as follows. With a further abuse of notation, in the sequel we will write $p^{-m}$ with $m\in\N$ even if this quantity might not exist in $\zetapk$. However, we will use it only when  it multiplies $p^{m'}$ for some integer $m'>m$. In such a way $p^{m'-m}$ is well-defined in $\zetapk$ and we will note it as product $p^{-m} \cdot p^{m'}$.  

Let $(G^{\Z},F)$ be any Additive CA and let $\locrule: G^{2r+1}\to G$ be its local rule defined, by $2r+1$ endomorphisms $\locrule_{-r}, \ldots, \locrule_{r}$ of $G$. For each $z\in\{-r, \ldots, r\}$, let $A_z = (a^{(z)}_{i,j})_{1\leq i\leq n,\, 1\leq j\leq n}\in (\zetapku)^{n\times n}$ be the matrix such that 
$
\forall i,j\in\{1, \ldots, n\}, a^{(z)}_{i,j}=p^{k_j-k_i} \cdot \locrule_z(e_j)^i
$. 
The \emph{LCA associated with the Additive CA $(G^{\Z},F)$} is $(\hat{G}^\Z, L)$, where $L$ is defined by $A_{-r}, \ldots, A_r$ or, equivalently, by $A=\sum_{z=-r}^r A_z X^{-z}\in\LPKuno^{n\times n}$.
We stress that the following diagram commutes
\[
\begin{CD}
   G^{\Z}@>F>>&G^{\Z}\\
   @V{\Emb}VV&@VV{\Emb}V\\
   \hat{G}^{\Z}@>>L>&\hat{G}^{\Z}
\end{CD}\enspace, 
\]
\ie, $L\circ \Emb=\Emb\circ F$. Therefore,  $(\hat{G}^\Z, L)$ is said to be the LCA associated with $(G^{\Z},F)$ \emph{via the embedding $\Emb$}. In general, $(G^{\Z},F)$ is not topologically conjugated (\ie, homeomorphic) to $(\hat{G}^\Z, L)$ but $(G^{\Z},F)$ is a subsystem of  $(\hat{G}^\Z, L)$ and the latter condition alone is not  enough in general to lift dynamical properties from a one system to the other one. Despite this obstacle, in the sequel we will succeed in doing such a lifting, as far as positive expansivity is concerned.

\section{Results}
\begin{definition}[Positive and Negative Degree]\label{deg+}
The \emph{positive} (resp., \emph{negative}) degree of any given polynomial $\coeff \in \LP_m$ with $\coeff\neq 0$, denoted by $\deg^+(\coeff)$ (resp., $\deg^-(\coeff)$), is the maximum (resp, minimum) value among the degrees of the monomials of $\coeff$. Such notions extend to any element $\upsilon\neq 0$ of $\LS^n_m$ when $\upsilon$ is considered as a formal power series with coefficients in $(\Z/m\Z)^n$ instead of a vector of $n$ elements from $\LS_m$ and with the additional defining clause that $\deg^+(\upsilon)=+\infty$ (resp., $\deg^-(\upsilon)=-\infty$) if that maximum (resp., minimum) does not exist. Furthermore, the previous notions are extended to both $\coeff=0$ and $\upsilon= 0$ as follows: $\deg^+(0)=-\infty$ and $\deg^-(0)=+\infty$.
\end{definition}
\begin{example}
The following are the values of the positive and negative degree of some polynomials:
$\deg^+(X^{-3}+X^{-2})=-2$,
$\deg^+(X^{-3}+X^{-2}+1)=0$, $\deg^+(X^{-3}+X^{-2}+1+X^{4})=4$,
$\deg^+(1)=0$\\
$\deg^-(X^{3}+X^{2})=2$,
$\deg^-(X^{3}+X^{2}+1)=0$, $\deg^-(X^{-3}+1+X^{4})=-3$, 
$\deg^-(1)=0$
\end{example}

\begin{definition}[Expansive Polynomial and Expansive Matrix]\label{exppoly}
 Let $\pol(t)=\coeff_0+ \coeff_1 t +\cdots + \coeff_{n-1} t^{n-1} + t^{n}$ be any polynomial from $\LP_m[t]$. We say that $\pol(t)$ is \emph{expansive} if  $\coeff_0\neq 0$ and both the following two conditions are satisfied: 
\begin{itemize}
\item[$(i)$] $\deg^+(\coeff_0)>0$ and 
$\deg^+(\coeff_0)>\deg^+(\coeff_i)$ for every $i \in \{1,\dots , n-1\}$;
\item[$(ii)$]  $\deg^-(\coeff_0)<0$ and $\deg^-(\coeff_0)<\deg^-(\coeff_i)$  for every $i \in \{1,\dots , n-1\}$;
\end{itemize}
A matrix $A\in \LP_m^{n\times n}$ is said to be \emph{expansive} if its characteristic polynomial is expansive.
\end{definition}
\begin{lemma}\label{pura}
 For any three polynomials $\pi, \rho, \tau\in \LP_p[t]$ such that $\pi=\rho\cdot \tau$ it holds that $\pi$ is expansive if and only if  $\rho$ and $\tau$ are both expansive.
  \end{lemma}
\begin{proof}
Choose arbitrarily three polynomials 
\begin{align*}
\pi(t)&=\alpha_0+ \alpha_1 t +\cdots + \alpha_{n-1} t^{n-1} + t^{n}\\
\rho(t)&=\beta_0+ \beta_1 t +\cdots + \beta_{n_1-1} t^{n_1-1} + t^{n_1}\\
\tau(t)&=\gamma_0+ \gamma_1 t +\cdots + \gamma_{n_2-1} t^{n_2-1} + t^{n_2}
\end{align*}
such that $\pi=\rho\cdot \tau$. Obviously,
\begin{equation}
\label{alpha}
\alpha_i=\sum_{i_1\leq i, i_2\leq i:  i=i_1+i_2} \beta_{i_1} \gamma_{i_2}
\end{equation}
for every $i\in\{0,\ldots, n-1\}$ and, in particular, $\alpha_0=\beta_0\gamma_0$.

Assume that both $\rho$ and $\tau$ are both expansive. Hence, $\alpha_0=\beta_0\gamma_0\neq 0$ and  
\[
\deg^+(\alpha_0)=\deg^+(\beta_0\gamma_0)= \deg^+(\beta_0) + \deg^+(\gamma_0) > \deg^+(\beta_{i_1}) + \deg^+(\gamma_{i_2})= \deg^+(\beta_{i_1} \gamma_{i_2})\enspace, 
\]
for every $i_1\in\{1,\ldots, n_1-1\}$ and $i_2\in\{1,\ldots, n_2-1\}$ such that $\beta_{i_1}\neq 0$ and $\gamma_{i_2}\neq 0$ (by Definition~\ref{exppoly} a symmetric inequality regarding $\deg^-$ also holds). Let $i\in\{1,\ldots, n-1\}$ be any index such that $\alpha_{i}\neq 0$. We can  write 
\[
\deg^+(\alpha_i)\leq \deg^+(\beta_{i_1} \gamma_{i_2})< \deg^+(\alpha_0)
\]
for every $i_1\in\{1,\ldots, n_1-1\}$ and $i_2\in\{1,\ldots, n_2-1\}$ such that $i=i_1+i_2$, $\beta_{i_1}\neq 0$ and $\gamma_{i_2}\neq 0$. Clearly, condition $(ii)$ from Definition~\ref{exppoly} also holds, as far as $\deg^-(\alpha_i)$ is concerned. Thus, $\pi$ is expansive.

We prove now that if $\pi$ is expansive then $\rho$ and $\tau$ are both expansive. 
Assume that the consequent is not true. We deal with the two following cases: (1) $\rho$ is expansive but $\tau$ is not (by symmetry, we do not consider the situation in which $\tau$ is expansive but $\rho$ is not); (2) neither $\rho$ nor $\tau$ are expansive. 
\begin{enumerate}
\item[(1)] Suppose that $\rho$ is expansive but $\tau$ is not. If $\gamma_0=0$ then it trivially follows that $\pi$ is not expansive. Otherwise, 
let $min$ be the minimum index such that $deg^+(\gamma_{min})=\max_{0\leq i_2 < n_2} \{deg^+(\gamma_{i_2}): \gamma_{i_2}\neq 0\}$. Since $\rho$ is expansive, $\beta_0\gamma_{min}$ is the (only) addend of maximum degree in the sum from Equation~\eqref{alpha} considered for $i=min$. Thus, $deg^+(\alpha_{min})= deg^+(\beta_0 \gamma_{min})$. Furthermore, $deg^+(\alpha_{min})= deg^+(\beta_0) + \deg^+(\gamma_{min})\geq  deg^+(\beta_0) + \deg^+(\gamma_0)= deg^+(\beta_0 \gamma_0)=deg^+(\alpha_0)$. In a symmetric way, one also gets that $deg^-(\alpha_{i})\leq  deg^-(\alpha_0)$ for some $i\in\{1, \ldots, n-1\}$. 
\item[(2)] Suppose that neither $\rho$ nor $\tau$ are expansive. If  $\beta_0=0 \vee\gamma_0=0$ then it trivially follows that $\pi$ is not expansive. Otherwise, 
let $max_1$ and $max_2$ be the maximum indexes such that $\deg^+(\beta_{max_1})=\max_{0\leq i_1 < n_1} \{\deg^+(\beta_{i_2}): \beta_{i_2}\neq 0\}$ and  $\deg^+(\gamma_{max_2})=\max_{0\leq i_2 < n_2} \{\deg^+(\gamma_{i_2}): \gamma_{i_2}\neq 0\}$, respectively.
Consider now Equation~\ref{alpha} for $i=max_1+max_2$. 
Take any pair of indexes $i_1, i_2$ of the sum such that $i_1\neq max_1$, $i_2\neq max_2$, $\beta_{i_1}\neq 0$, and $\gamma_{i_2}\neq 0$. 
If $i_1<max_1$ (and then $i_2>max_2$), resp.,  if $i_1>max_1$ (and then $i_2<max_2$), we get that 
$$
\deg^+(\beta_{i_1}) + \deg^+(\gamma_{i_2})< \deg^+(\beta_{i_1}) + \deg^+(\gamma_{max_2})\leq \deg^+(\beta_{max_1}) + \deg^+(\gamma_{max_2})\enspace,
$$
resp.,
$$
\deg^+(\beta_{i_1}) + \deg^+(\gamma_{i_2})< \deg^+(\beta_{max_1}) + \deg^+(\gamma_{i_2})\leq \deg^+(\beta_{max_1}) + \deg^+(\gamma_{max_2})\enspace. 
$$ Hence,  $\deg^+(\beta_{i_1}\gamma_{i_2})<\deg^+(\beta_{max_1}\gamma_{max_2})$. This implies that $deg^+(\alpha_{max_1+max_2})=\deg^+(\beta_{max_1}\gamma_{max_2})$. Moreover, $deg^+(\alpha_{max_1+max_2})=\deg^+(\beta_{max_1}\gamma_{max_2})\geq \deg^+(\beta_{0}\gamma_{0})=\deg^+(\alpha_{0})$ and in a symmetric way, one also gets that $deg^-(\alpha_{i})\leq  deg^-(\alpha_0)$ for some $i\in\{1, \ldots, n-1\}$.
\end{enumerate}
Therefore, in both cases it follows that $\pi$ is not expansive and this concludes the proof.
\end{proof}
\begin{definition}\label{sk}
Let $K=\LS^n_m$. 
For any $s\in\Z$ we define the following sets:

$Right(K,s)=\{\upsilon\in K: deg^-(\upsilon)= s \}$,

$Right^*(K,s)=\{\upsilon\in K: deg^-(\upsilon)\geq s \}$,

$Left(K,s)=\{\upsilon\in K: deg^+(\upsilon)= s \}$,

$Left^*(K,s)=\{\upsilon\in K: deg^+(\upsilon)\leq s \}$.
\end{definition}
Definition~\ref{sk} along with the notion of positively expansive CA when reformulated for LCA and the compactness of the configuration space immediately allow stating the following  

\begin{lemma}\label{expansivity}
Let $\glorule$ be any LCA over $(\Z/m\Z)^n$ and let  $A\in\LP_{m}^{n\times n}$ be the matrix associated with $\glorule$, where $m$ and $n$ are any two naturals with $m>1$ and $n>1$. 
The LCA $\glorule$ is not positively expansive if and only if there exists an integer  $s>0$ such that at least one of the following two conditions holds
\begin{enumerate}
\item[-]
$\exists \upsilon\in Left(\LS^n_m,0)\setminus\{0\}$:
$\forall \ell>0, \,  A^{\ell}\upsilon \in Left^*(\LS^n_m,s)$
\item[-] $\exists \upsilon\in Right(\LS^n_m,0)\setminus\{0\}$:  
$\forall \ell>0, \, A^{\ell}\upsilon \in Right^*(\LS^n_m,-s)$ 
 \end{enumerate}
On the contrary, the LCA $\glorule$ is positively expansive if and only if there exists a natural number  
$\hat{\ell}>0$ such that for any $d\in\Z$ and any $\upsilon\in Left(\LS^n_m,d)$ it holds that $A^{\ell}\upsilon \in Left(\LS^n_m,d+1)$ for some $\ell\leq \hat{\ell}$ and, symmetrically, for any $d\in\Z$ and any $\upsilon\in Right(\LS^n_m,d)$ it holds that $A^{\ell}\upsilon \in Right(\LS^n_m,d-1)$ for some $\ell\leq \hat{\ell}$.
\end{lemma}

\begin{lemma}\label{branch}
Let $\glorule$ be any surjective LCA over $(\Z/m\Z)^n$ and let  $A\in\LP_{m}^{n\times n}$ be the matrix associated with $\glorule$, where $m$ and $n$ are any two naturals with $m>1$ and $n>1$. 
Then there exists an integer constant $c\geq 0$ 
(that only depends on $A$) such that all the following conditions $\mathcal{C}_1, \mathcal{C}_2,\mathcal{C}_3$, and $\mathcal{C}_4$ hold.

\begin{enumerate}
\item[$\mathcal{C}_1$:] $\forall \upsilon\in Left(\LS^n_m,0)\ \exists \omega\in Left(\LS^n_m,h)$ for some $-c\leq h \leq c$ such that  $A \omega=\upsilon$ %
\item[$\mathcal{C}_2$:] $\forall \upsilon\in Right(\LS^n_m,0)\ \exists \omega\in Right(\LS^n_m,h)$ for some $-c\leq h \leq c$ such that  $A \omega =\upsilon$ 
\item[$\mathcal{C}_3$:] $\forall \upsilon\in Left(\LS^n_m,0),\, A \upsilon \in Left(\LS^n_m,h)$ for some $-c\leq h \leq c$
\item[$\mathcal{C}_4$:] $\forall \upsilon\in Right(\LS^n_m,0),\, A \upsilon \in Right(\LS^n_m,h)$ for some $-c\leq h \leq c$
\end{enumerate}
\end{lemma}

\begin{proof}
It is an immediate consequence of the fact that $\glorule$ is open and, hence, it is both left and right closing. 
\end{proof}
We now state the result that is the heart of our work, \ie, providing a decidable characterization of positive expansivity for LCA over $(\Z/p\Z)^n$. Since its proof is very long, we place it in Section~\ref{mainpr}.
\begin{theorem}
\label{risp}
Let $\glorule$ be any LCA over $(\Z/p\Z)^n$ where $p$ and $n$ are any two naturals such that $p$ is prime and $n>1$. The LCA $\glorule$ is positively expansive if and only if the matrix associated with $\glorule$ is expansive.
\end{theorem}
The following Lemma extends the decidability result provided by Theorem~\ref{risp} from LCA over $(\Z/p\Z)^n$ to LCA over $(\Z/p^k\Z)^n$.
\begin{lemma}
\label{rispk}
Let $\mathcal{G}$ be any LCA over $(\Z/p^k\Z)^n$ and let  $B\in\LP_{p^k}^{n\times n}$ be the matrix associated with $\mathcal{G}$, where $p, k, n$ are any three naturals such that $p$ is prime, $k>1$, and $n>1$. The LCA $\mathcal{G}$ is positively expansive if and only if the LCA $\glorule$ over $(\Z/p\Z)^n$ having $A$ as associated matrix is too, where $A=(B\mmod p)\in\LP_{p}^{n\times n}$. Equivalently, $\mathcal{G}$ is positively expansive if and only if the matrix $B\mmod p$ is expansive.
\end{lemma}
\begin{proof}
We start to prove that if $\mathcal{G}$ is positively expansive then $\glorule$ is too. Let us suppose  that $\glorule$ is not positively expansive and the first condition from Lemma~\ref{expansivity} holds, \ie,  there exists $\upsilon\in Left(\LS^n_p,0)$ with $\upsilon\neq 0$ such that $A^{\ell}\upsilon \in Left^*(\LS^n_p,s)$ for every natural $\ell>0$, where $s>0$ is some integer constant depending on $\glorule$ (the proof is symmetric if one supposes that the second condition holds). Set $\omega=p^{k-1}\upsilon$. Clearly, $\omega\in Left(\LS^n_{p^k},0)$ and $\omega\neq 0$. Since $B$ can be written as $A+pN$ for some matrix $N\in\LP_{p^{k-1}}^{n\times n}$, it holds that  for every natural $\ell>0$
\[
B^{\ell}\omega=(A+pN)^{\ell}p^{k-1}\upsilon=p^{k-1} A^{\ell}\upsilon\enspace,
\]
where all the sums and products of coefficients of the Laurent polynomials/series inside the previous equalities are now meant in $\Z_{p^k}$. Hence, although $A\in\LP_{p}^{n\times n}$ and $\upsilon\in \LS_p^n$, in general $A^{\ell}\upsilon$ now belongs to $\LS^n_{p^k}$ instead of $\LS^n_{p}$ and, as a consequence,  we can not immediately conclude  that $B^{\ell}\omega= p^{k-1} A^{\ell}\upsilon\in Left^*(\LS^n_{p^k},s)$ immediately follows just from the hypothesis as it is, \ie, $A^{\ell}\upsilon \in Left^*(\LS^n_p,s)$ when $A^{\ell}\upsilon$ is considered as an element of $\LS^n_{p}$. However, since $A^{\ell}\upsilon\in\LS^n_{p^k}$ can be written as  $(A^{\ell}\upsilon)\mmod p + p\omega'$ for some $\omega'\in \LS^n_{p^{k-1}}$ and, by hypothesis, $(A^{\ell}\upsilon)\mmod p\in Left^*(\LS^n_p,s)$, we get that $B^{\ell}\omega= p^{k-1} [ (A^{\ell}\upsilon)\mmod p + p\omega']=  p^{k-1} [(A^{\ell}\upsilon)\mmod p]\in Left^*(\LS^n_{p^k},s)$ for every natural $\ell>0$. Therefore, by Lemma~\ref{expansivity}, $\mathcal{G}$ is not positively expansive.


We now prove that if $\glorule$  is positively expansive then $\mathcal{G}$ is too. Let us suppose  that $\mathcal{G}$ is not positively expansive and the first condition from Lemma~\ref{expansivity} holds, \ie,  there exists $\omega\in Left(\LS^n_{p^k},0)$ with $\omega\neq 0$ such that $B^{\ell}\omega \in Left^*(\LS^n_{p^k},s)$ for every natural $\ell>0$, where $s>0$ is some integer constant depending on $\mathcal{G}$ (again, the proof is symmetric if one supposes that the second condition holds). We deal with the following two mutually exclusive cases. 
\begin{itemize}
\item[]
If there is no $\omega'$ such that $\omega=p\omega'$, then set $\upsilon=\omega\mmod p$. Clearly, $\upsilon\neq 0$ and  $\upsilon \in Left(\LS^n_{p},h)$ for some integer $h\leq 0$. Moreover, it holds that $A^{\ell}\upsilon = (B\mmod p)^{\ell} (\omega \mmod p) = (B^{\ell}\omega)\mmod p \in Left^*(\LS^n_{p},s)$ for every natural $\ell>0$. 
\item[]
Otherwise, let $j\in\{1,\ldots, k-1\}$ and $\omega'\in\LS^n_{p^{k}}$ be such that $\omega=p^j\omega'$, where $j$ is the largest natural such that all the coefficients of the $n$ Laurent series forming $\omega$ are multiple of $p^j$.  Set $\upsilon=\omega'\mmod p$. Clearly, $\upsilon\neq 0$ and  $\upsilon \in Left(\LS^n_{p},h)$ for some integer $h\leq 0$. Since  $ p^j B^{\ell}\omega'=B^{\ell}\omega\in Left^*(\LS^n_{p^k},s)$, either $B^{\ell}\omega'\in Left^*(\LS^n_{p^k},s)$ or $B^{\ell}\omega'\notin Left^*(\LS^n_{p^k},s)$ happens, but, in both situations it must hold that $(B^{\ell}\omega') \mmod p \in Left^*(\LS^n_{p},s)$. Therefore, it follows that $A^{\ell}\upsilon = (B\mmod p)^{\ell} (\omega' \mmod p) = (B^{\ell}\omega')\mmod p \in Left^*(\LS^n_{p},s)$ for every natural $\ell>0$.
\end{itemize}
In both cases, the first condition of Lemma~\ref{expansivity} is satisfied as far as $\glorule$ is concerned.  Thus, $\glorule$ is not positively expansive and this concludes the proof that $\mathcal{G}$ is positively expansive if and only if $\glorule$ is positively expansive. By Theorem~\ref{risp}, it follows that $\mathcal{G}$ is positively expansive if and only if the matrix $B\mmod p$ is expansive.
\end{proof}
At this point, we are able to extend the decidability result regarding positive expansivity to the whole class of LCA over $(\Z/m\Z)^n$ where $m$ is any natural with $m>1$. 
\begin{corollary}
\label{rism}
Positive expansivity is decidable for Linear CA over $(\Z/m\Z)^n$.
\end{corollary}
\begin{proof}
Consider an arbitrary LCA $\glorule$ over $(\Z/m\Z)^n$ and let $A\in {\LP_m}^{n\times n}$ be the matrix associated with $\glorule$. Let  $m=p_1^{k_1}\cdots p_l^{k_l}$ be the prime factor decomposition of $m$ and, for each $i\in\{1, \ldots, l\}$, let $\glorule_i$ be the LCA over $(\Z/(p_i)^{k_i}\Z)^n$ having $(A \mmod (p_i)^{k_i}) \in {\LP^{n\times n}_{(p_i)^{k_i}}}$ as associated matrix. Since $\glorule$ is positively expansive if and only if every LCA $\glorule_i$ is too, by Lemma~\ref{rispk} and Theorem~\ref{risp}, it follows that $\glorule$ is positively expansive if and only if every matrix $(A \mmod p_i) \in {\LP}_{p_i}^{n\times n}$ is expansive, the latter being a decidable property. Therefore, the statement is true.
\end{proof}
Finally, we prove the decidability result regarding positive expansivity for the whole class Additive CA  over a finite abelian group. 
\begin{corollary}
Positive expansivity is decidable for Additive CA over a finite abelian group.
\end{corollary}
\begin{proof}
Let $\glorule: G^\Z\to G^\Z$ be any Additive CA over the finite abelian group. Without loss of generality, we can assume that $G=\zetapku\times\ldots\times \zetapkn$ for some prime $p$ and some non zero naturals $k_1, \ldots, k_n$ with $k_1\geq k_2\geq \ldots\geq k_n$. Let $\mathcal{L}$ be the LCA over $\hat{G}$ associated with $\glorule$ via the embedding $\Emb$, where $\hat{G}=(\zetapku)^n$. We are going to show that $\glorule$ is positively expansive if and only if $\mathcal{L}$ is too. By Corollary~\ref{rism}, this is enough to conclude that positive expansivity is decidable for Additive CA over a finite abelian group.

We start to prove that if $\mathcal{L}$ is positively expansive then $\glorule$ is too. Let us suppose that $\glorule$ is not positively expansive. Choose an arbitrary $\varepsilon>0$. So, there exist $c, c'\in G^\Z$ with $c\neq c'$ such that $d({\glorule}^\ell(c), {\glorule}^\ell(c'))\leq \varepsilon$ for every natural $\ell$. Consider the two configurations $\Emb(c), \Emb(c')\in \hat{G}^\Z$. Clearly, $\Emb(c)\neq \Emb(c')$. Moreover, for every natural $\ell$ it holds that $d({\mathcal{L}}^\ell(\Emb(c)), {\mathcal{L}}^\ell(\Emb(c')))=d(\Emb({\glorule}^\ell(c)), \Emb({\glorule}^\ell(c')))=d({\glorule}^\ell(c), {\glorule}^\ell(c'))\leq\varepsilon$. Hence, $\mathcal{L}$ is not positively expansive.

We now prove that if $\glorule$ is positively expansive then $\mathcal{L}$  is too. Let us suppose that $\mathcal{L}$ is not positively expansive. Choose an arbitrary $\varepsilon>0$. So, there exist $b, b'\in \hat{G}^\Z$ with $b\neq b'$ such that $d({\mathcal{L}}^\ell(b), {\mathcal{L}}^\ell(b'))\leq \varepsilon$ for every natural $\ell$. Let $min$ be the minimum natural such that $p^{min} \cdot (b-b')\neq 0$ and $p^{min+1}\cdot (b-b')=0$. We get that $p^{min} \cdot b, p^{min} \cdot b'\in \Emb(G^\Z)$ and $p^{min} \cdot b\neq p^{min} \cdot b'$. Let $c,c'\in G^\Z$ be the two configurations such that $\Emb(c)=p^{min} \cdot b$ and $\Emb(c')=p^{min} \cdot b'$. Clearly, $c\neq c'$. 
For every natural $\ell$ it holds that $d({\glorule}^\ell(c), {\glorule}^\ell(c')) =d(\Emb({\glorule}^\ell(c)), \Emb({\glorule}^\ell(c')))=d({\mathcal{L}}^\ell(\Emb(c)), {\mathcal{L}}^\ell(\Emb(c')))=d({\mathcal{L}}^\ell(b), {\mathcal{L}}^\ell(b'))\leq\varepsilon$. Hence $\glorule$ is not positively expansive.
\end{proof}
\section{Proof of Theorem~\ref{risp}}
\label{mainpr}
We now recall some useful notions and known fact from abstract algebra. In the sequel, the standard acronyms  PID and UFD stand for principal ideal domain and unique factorization domain, respectively. 
\smallskip

Let $\pid$ be PID and let $A\in \pid^{n\times n}$. The \emph{elementary divisors}, or \emph{invariants}, or \emph{invariant factors} associated with $A$ are the elements $a_1,\dots, a_n\in\pid$ defined as follows:  
$\forall i\in \{1,\dots,n \}, a_i = \Delta_i(A)/\Delta_{i-1}(A)$,  where $\Delta_i(A)$ is the greatest common divisor of the $i$-minors of $A$ and $\Delta_0(A) = 1$.
\smallskip

The companion matrix of a monic polynomial $\pol(t)=\coeff_0 +  \ldots +\coeff_{n-1} t^{n-1}+ t^n$ is 

$$C_{\pi}=
\begin{pmatrix}
0&0&0&-\coeff_0\\
1&\cdots&0&-\coeff_1\\
\vdots&\ddots&\vdots&\vdots\\
0&\cdots&1&-\coeff_{n-1}
\end{pmatrix}\enspace.
$$
A matrix $C$ is in rational canonical form if it is block diagonal
$$C=
\begin{pmatrix}
 C_{\pi_1}  & \mathbb{O} & \cdots & \mathbb{O} \\
  \mathbb{O} &C_{\pi_2}  &\cdots  &  \mathbb{O} \\
  \vdots &\vdots  &\ddots  & \vdots \\
 \mathbb{O}  & \mathbb{O} & \cdots &C_{\pi_s}  \\
\end{pmatrix}\enspace,
$$
where each $C_{\pi_i}$ is the companion matrix of some monic polynomial $\pi_i$ of non null degree and $\pi_i$ divides $\pi_j$ for $i \leq j$. It is well-known that if $A\in \field^{n\times n}$ is any matrix  where $\field$ is a field all the following facts hold:
\begin{enumerate}
\item[-] $A$ is similar to a  unique matrix in rational canonical form and this latter is called the rational canonical form of $A$;
\item[-] the monic polynomials $\pi_1, \ldots, \pi_s$ defining the blocks of the rational canonical form of $A$ are the nonconstant invariant factors of $tI-A$;
\item[-] $\chi_A(t)=\prod_{i=1}^{s} \pi_i(t)$, where $\pi_s$ is the minimal polynomial of $A$;
\item[-] there exist
$v_1,\dots, v_s\in \field^n$ such that 
$$\{v_1,Av_1,\dots,A^{d_1-1}v_1,     v_2,Av_2,\dots,A^{d_2-1}v_2,\dots, v_s,Av_s,\dots,A^{d_s-1}v_s\}$$
is a basis of $ \field^n$ with respect to which $A$ becomes in rational canonical form $C$, \ie, $A=P^{-1} C P$, where $d_i=\deg(\pi_i)$ and $P$ is the matrix having the elements of that basis as columns.
\end{enumerate}
We now report the following known result which will be very useful in the sequel.
\begin{lemma}[Proposition 1 in https://people.math.binghamton.edu/mazur/teach/gausslemma.pdf]
\label{lemmafigo}
Let $U$ be a UFD and let $\field_U$ be the field of fractions of $U$. For any monic polynomials $\pi \in U[t]$ and $\rho, \tau\in \field_U[t]$, it holds that if $\pi = \rho\cdot \tau$ then $\rho,\tau \in U[t]$.
\end{lemma}
The following is an important consequence of Lemma~\ref{lemmafigo}. 
\begin{lemma}\label{kf}
Let $U$ be a UFD and let $\field_U$ be the field of fractions of $U$. Let $A=U^{n\times n}$ and let $\pi_1,\ldots, \pi_s\in \field_U[t]$ be the invariant factors of $tI-A$ when $A$ is considered as an 
element of $\field_U^{n\times n}$. Then, for every $i\in\{1, \ldots, s\}$ it holds that $\pi_i \in U[t]$ .
\end{lemma}
\begin{proof}
Since $\pi_1,\ldots, \pi_s, \chi_A$ are all monic and $\chi_A(t)=\prod_{i=1}^{s} \pi_i(t)\in U[t]$, by a repeated application of Lemma~\ref{lemmafigo}, we get that every $\pi_i\in U[t]$.
\end{proof}
We now deal with the algebraic structures of our interest, namely, the PID $\LP_p$ and the UFD $\LP_p[t]$. Clearly, $\LP_p$ is also an UFD, but, since $\LP_p$ is not a field, $\LP_p[t]$ is not a PID, while $\F_p[t]$ is, where $\F_p$ is the field of fraction of  $\LP_p$. Therefore, we can not refer to invariant factors when  the involved set is $\LP_p[t]$, while we can as far as $\F_p[t]$ is concerned.
%
\begin{lemma}\label{main1}
For any matrix $A\in \LP_p^{n\times n}$ with $\det(A)\neq 0$ there exist two matrices $Q,C\in \LP_p^{n\times n}$ such that $\det(Q)\neq 0$, $C$ is in rational canonical form, $QA=CQ$, and $\chi_A=\chi_C$.
%
\end{lemma}
\begin{proof} 
Choose arbitrarily a matrix $A\in \LP_p^{n\times n}$. Cleary, it holds that $A\in \F_p^{n\times n}$, where $\F_p$ is the  field of fractions of $\LP_p$. Hence,  there exist $v_1,\dots, v_s\in \field^n$ such that, $A=P^{-1} C P$, $C\in \field_p^{n\times n}$ is the matrix in canonical rational form, the blocks of which are defined by the invariant factors $\pi_1, \ldots, \pi_s\in \F_p[t]$ of $tI-A$,  and $P$ is the matrix having the elements of the basis $\{v_1,Av_1,\dots,A^{d_1-1}v_1,     v_2,Av_2,\dots,A^{d_2-1}v_2,\dots, v_s,Av_s,\dots,A^{d_s-1}v_s\}$ of $ \field_p^n$ as columns, where $d_i=\deg(\pi_i)$. Clearly, $\chi_A=\chi_C$.

Let $\alpha\in \LP_p$ be such that $v_i'=\alpha v_i\in \LP_p$ for each $i\in \{1,\dots,s \}$. Set $Q=\alpha P$. It is clear that $Q\in \LP_p^{n\times n}$, as desired. Furthermore, by Lemma~\ref{kf}, it follows that $C\in \LP_p^{n\times n}$, too. Moreover, we get that $A=P^{-1} C P=\alpha \alpha^{-1}P^{-1} C P= \alpha^{-1}P^{-1} C \alpha P=Q^{-1}CQ$. Therefore, $QA=CQ$ and this concludes the proof. 
\end{proof}

\begin{lemma}\label{main2}
For any $A,B,Q\in \LP_p^{n\times n}$ such that $\det(Q)\neq 0$ and $AQ=QB$ it holds that the LCA $\glorule$ over $(\Z/p\Z)^n$ having $A$ as associated matrix is positively expansive if and only if the LCA $\mathcal{G}$ over $(\Z/p\Z)^n$ having $B$ as associated matrix is, too.
\end{lemma}
\begin{proof}
First of all, it is easy to see that $A^{\ell}Q=QB^{\ell}$ for every natural $\ell>0$. Moreover, $\det(A)=0$ if and only if $\det(B)=0$. So, the thesis turns out to be trivially true if $\det(A)=\det(B)=0$. Therefore, in the sequel of the proof, we will assume that $\det(A)\neq 0$ and $\det(B)\neq 0$, \ie, both $\glorule$ and $\mathcal{G}$ are surjective.

We now start to prove that if $\glorule$ is positively expansive then $\mathcal{G}$ is too. Suppose that $\mathcal{G}$ is not positively expansive and the first condition from Lemma~\ref{expansivity} holds, \ie,  there exists $\upsilon\in Left(\LS^n_{p},0)$ with $\upsilon\neq 0$ such that $B^{\ell}\upsilon\in Left^*(\LS^n_{p},s)$ for every natural $\ell>0$, where $s>0$ is some integer constant depending on $\mathcal{G}$ (the proof is symmetric if one supposes that the second condition holds). Since $Q$ is the matrix associated with a surjective LCA over $(\Z/p\Z)^n$ condition $\mathcal{C}_3$ of  Lemma~\ref{branch} is satisfied as far as $Q$ is concerned. 
%
Hence, setting $\omega=Q\upsilon$, for some natural constant $c>0$ depending on $Q$ and some integer $h$ with $-c\leq h\leq c$, it holds that  $A^{\ell}\omega=QB^{\ell}\upsilon\in Left^*(\LS^n_{p},h)$ for every natural $\ell>0$. Clearly, 
$\omega\neq 0$ since $det(Q)\neq 0$. In addition, it holds that  $\omega\in Left(\LS^n_{p},h')$ for some integer $h'$. Thus, it follows that $\glorule$ is not positively expansive.

We now prove that if $\mathcal{G}$  is positively expansive then $\glorule$ is too. Assume that $\glorule$ is not positively expansive and the first condition from Lemma~\ref{expansivity} holds, \ie,  there exists $\upsilon\in Left(\LS^n_{p},0)$ with $\upsilon\neq 0$ such that $A^{\ell}\upsilon\in Left^*(\LS^n_{p},s)$ for every natural $\ell>0$, where $s>0$ is some integer constant depending on $\mathcal{G}$ (again, the proof is symmetric if one supposes that the second condition holds). 
Since $Q$ is the matrix associated with a surjective LCA over $(\Z/p\Z)^n$, condition $\mathcal{C}_1$ of  Lemma~\ref{branch} is satisfied as far as $Q$ is concerned. 
Thus, for some natural constant $c>0$ depending on $Q$ and some integer $h$ with $-c\leq h\leq c$, there exists $\omega\in Left(\LS^n_{p},h)$ such that $Q\omega=\upsilon$. Clearly, $\omega\neq 0$ since $\upsilon\neq 0$. Furthermore, $Q B^{\ell}\omega=A^{\ell}Q\omega=A^{\ell}\upsilon\in Left^*(\LS^n_{p},s)$ for every natural $\ell>0$. Therefore, 
there exists an integer constant $s'>s>0$ such that $B^{\ell}\omega\in Left^*(\LS^n_{p},s')$ for every natural $\ell>0$. So, by Lemma~\ref{expansivity}$, \mathcal{G}$ is not positively expansive and this concludes the proof. 
%
%
%
\end{proof}
In the sequel, with an abuse of notation, $\deg^+(\varphi)$ stands for $\deg^+(\alpha)-\deg^+(\beta)$ for any fraction $\varphi=\alpha/\beta\in \F_p$ with $\alpha,\beta\in\LP_p$ and $\alpha,\beta\neq 0$, where $\F_p$ is the field of fractions of $\LP_p$. 
\begin{lemma}\label{lx}
Let $\upsilon_1,\dots,\upsilon_n$ be arbitrary elements of $\LP_p$, 
where $p$ and $n$ are any two naturals such that $p$ is prime and $n>1$, and let  $\varphi_1,\dots,\varphi_n$ be arbitrary elements of $\F_p$ such that 
$\max \{ \deg^+(\varphi_1),\ldots,\deg^+(\varphi_n)\}\geq 0$. Let $min$ be the minimum index such that 
 $\deg^+(\varphi_{min})= \max \{ \deg^+(\varphi_1),\ldots,\deg^+(\varphi_n)\}$. 
Regarding the sequence 
\[
\{\upsilon_1, \ldots, \upsilon_n, \ldots,  \upsilon_{j}, \ldots \}\subset\F_p\enspace,
\]
where, for each $j>n$, 
\[
\upsilon_{j}= \varphi_n\upsilon_{j-1}  
  + \ldots +
  \varphi_1\upsilon_{j-n}\enspace.
 \]
call pick any natural $J>0$ such that $deg^+(\upsilon_j)\leq deg^+(\upsilon_{J})$ for every natural $j$ with $0<j<J$. It holds that 
for any pick $J$ there exists a pick 
$J'\in\{J+1, \ldots, J+n-min+1\}$. In particular, for any pick the number of positions within which there is a further pick does not depend on the values of the initial elements $\upsilon_1, \ldots, \upsilon_n$ of the sequence.
  \end{lemma}
 \begin{proof}
Clearly, the set of picks is non empty since $1,\ldots, n$ are picks.  
Let $J$ be any pick. If there exists $j\in\{J+1, \ldots, J+n-min\}$ such that $deg^+(\upsilon_{j})\geq deg^+(\upsilon_{J})$, necessarily there must be a further pick inside the integer interval $\{J+1, \ldots, J+n-min\}$. Otherwise, it holds that  $deg^+(\upsilon_{j})< deg^+(\upsilon_{J})$ for every $j\in\{J+1, \ldots, J+n-min\}$ and, since
\[
 \deg^+(\upsilon_{J+n-min+1})=\deg^+(\upsilon_{J}\varphi_{min})= \deg^+(\upsilon_{J})+\deg^+(\varphi_{min})\geq \deg^+(\upsilon_{J})\enspace,
 \]
 it follows that the natural $J+n-min+1$ turns out to be pick. 
\end{proof} 
The following result is the heart of our work. It provides a decidable characterization of positively expansive LCA over $(\Z/p\Z)^n$ with  associated matrix such that its transpose is in a rational canonical form consisting of only one block.
\begin{lemma}\label{exr}
Let $A\in\LP_p^{n\times n}$ be the matrix such that its transpose $A^T$ is the companion matrix of any monic polynomial 
$-\coeff_0 +  \ldots - \coeff_{n-1} t^{n-1}+ t^n$
from $\LP_p[t]$, \ie,
\begin{equation}
\label{riga}
A=
\begin{pmatrix}
0&1&0&\cdots&0\\
0&0&1&\ddots&0\\
\vdots&\ddots&\ddots&\ddots&0\\
0&0&\cdots&0&1\\
\coeff_0&\coeff_1&\cdots&\coeff_{n-2}&\coeff_{n-1}
\end{pmatrix}\enspace,
\end{equation}
where $p$ and $n$ are any two naturals such that $p$ is prime and $n>1$, and let $\glorule$ be the LCA over $(\Z/p\Z)^n$ having $A$ as associated matrix. 
The LCA $\glorule$ is positively expansive if and only if $A$ is expansive if and only if $A^T$ is expansive if and only if $\coeff_0 +  \ldots +\coeff_{n-1} t^{n-1}+ t^n$ is expansive. 
\end{lemma}
\begin{proof}
It is clear that the matrix $A$ is expansive if and only if its transpose $A^T$ is expansive if and only if $\chi_A(t)=-\coeff_0 +  \ldots -\coeff_{n-1} t^{n-1}+ t^n$ is expansive if and only if $\coeff_0 +  \ldots +\coeff_{n-1} t^{n-1}+ t^n$ is expansive. Since $\glorule$ is surjective if and only if $-\coeff_0\neq 0$, the thesis turns out to be trivially true if $\coeff_0=0$.  Hence, in the sequel of the proof we can assume that $\coeff_0\neq 0$.

We start to prove that if $A$ is expansive then $\glorule$ is positively expansive. For a sake of argument, suppose that  $A$ is expansive but $\glorule$ is not positively expansive and, in particular, the first condition from Lemma~\ref{expansivity} holds, \ie,  there exists $\upsilon=(\upsilon_1,\ldots, \upsilon_n)\in Left(\LS^n_{p},0)$ with $\upsilon\neq 0$ such that $A^{\ell}\upsilon\in Left^*(\LS^n_{p},s)$ for every natural $\ell>0$, where $s>0$ is some integer constant depending on $\glorule$ (the proof is symmetric if one supposes that the second condition holds). Consider the infinite sequence 
\[
\{\upsilon_1, \ldots, \upsilon_n, \upsilon_{1+n}, \ldots,  \upsilon_{\ell+n}, \ldots \}\enspace,
\]
where $\upsilon_{\ell+n}= \coeff_0 \upsilon_{\ell} +  \ldots +\coeff_{n-1} \upsilon_{\ell+n-1}$ for each natural $\ell>0$. Clearly, it holds that $A^{\ell}\upsilon=(\upsilon_{\ell+1}, \ldots, \upsilon_{\ell+n})$ for each $\ell>0$. The first condition from Lemma~\ref{expansivity} ensures that there exists an integer $s'\leq s$ such that $\deg^+(\upsilon_j)\leq s'$ for every natural $j>0$ and $\deg^+(\upsilon_{j})=s'$ for at least one $j>0$. Let $min>0$ be the minimum index such that $\deg^+(\upsilon_{min})=s'$.
Since
$$\upsilon_{min+n}= \coeff_0 \upsilon_{min} +  \ldots +\coeff_{n-1} \upsilon_{min+n-1}\enspace,$$
$A$ is expansive, and $p$ is prime, we get $\deg^+(\upsilon_{min+n})=\deg^+(\alpha_0 \upsilon_{min})>s'$, which contradicts that $\deg^+(\upsilon_j)\leq s'$ for every natural $j>0$. 

We now prove that if $\glorule$ is positively expansive then $A$ is expansive. 
Set
\[
\varphi_n= -\frac{\alpha_1}{\alpha_0}, \varphi_{n-1}= -\frac{\alpha_2}{\alpha_0}, \ldots, \varphi_{2}= -\frac{\alpha_{n-1}}{\alpha_0},\varphi_1= \frac{1}{\alpha_0}\enspace.
\]
Assume that $A$ is not expansive, \ie, equivalently, $\coeff_0 +  \ldots +\coeff_{n-1} t^{n-1}+ t^n$ is not  expansive, and condition $(i)$ from Definition~\ref{exppoly} does not hold, \ie, either $\deg^+(\alpha_0)\leq 0$ 
 or $\deg^+(\alpha_0)\leq \max\{\deg^+(\alpha_1), \ldots, \deg^+(\alpha_{n-1})\}$ (the proof is symmetric if one supposes that  condition  $(ii)$ does not hold). In both cases we get that  $\max \{ \deg^+(\varphi_1),\ldots,\deg^+(\varphi_n)\}\geq 0$. 
 Let $\F_p$ be the field of fraction of $\LP_p$. 
For any $\upsilon=(\upsilon_1,\ldots, \upsilon_n)\in \F_p^n$ define the sequence
\[
\{\upsilon_1, \ldots, \upsilon_n, \ldots,  \upsilon_{j}, \ldots \}\enspace,
\]
where, for each $j>n$, 
\[
\upsilon_{j}= \varphi_n \upsilon_{j-1}  +
\varphi_{n-1} \upsilon_{j-2} +\cdots + 
\varphi_2 \upsilon_{j-n+1} + 
\varphi_1 \upsilon_{j-n}\in \F_p\enspace.
\]
We emphasize that $A\,(\upsilon_j, \ldots, \upsilon_{j-n+1})=(\upsilon_{j-1}, \ldots, \upsilon_{j-n})$. The hypothesis of Lemma~\ref{lx} are satisfied and, hence, for every natural $j$ the integer interval $\{js+1, \ldots, (j+1)s\}$ contains a pick, where $s=n-min+1$ (with $min$ as in Lemma~\ref{lx}). We stress that $s$ does not depend on $(\upsilon_1, \ldots, \upsilon_n)$. For every natural $\ell>0$, we are now going to exhibit an integer $d^{(\ell)}$ and an element $\upsilon^{(\ell)}\in Left(\LP_p^n, d^{(\ell)})$ such that $A^{\ell'}\upsilon^{(\ell)}\in Left^*(\LP_p^n, d^{(\ell)})$ for each natural $\ell'\leq \ell$.  By Lemma~\ref{expansivity}, this is enough to state that $\glorule$ is not positively expansive and this concludes the proof.

So, to proceed, choose an arbitrary natural $\ell>0$ and let $h$ be such that $hs+1>\ell+n$. We know that $\{hs+1, \ldots, (h+1)s\}$ contains at least one pick whatever the first $n$ values $\upsilon_1, \ldots, \upsilon_n$ of the above sequence are (while the specific value of a pick inside that interval depends on the values of $\upsilon_1, \ldots, \upsilon_n$). Consider now   
\[
(\upsilon_1, \ldots, \upsilon_n)=\left(  \left(\coeff_0\right)^{h'},\ldots ,\left(\coeff_0\right)^{h'} \right)\enspace,
\]
where 
$h'=(h+1)s$. Regarding the above sequence when its first $n$ elements are just the components of such $(\upsilon_1, \ldots, \upsilon_n)$, let $J$ be the value of a pick inside $\{hs+1, \ldots, (h+1)s\}$. Set $\upsilon^{(\ell)}=(\upsilon_J, \ldots, \upsilon_{J-n+1})$ and let $d^{(\ell)}=\deg^+(\upsilon^{(\ell)})=\deg^+(\upsilon_J)$. At this point, we are able to state that  all the following facts hold:
\begin{enumerate}
\item[-] $\upsilon_j\in \LP_p$ for each natural $j$ with $0<j\leq h'$ and, hence, $A^{j}(\upsilon_{h'}, \ldots, \upsilon_{h'-n+1})\in\LP_p^n$ for each natural $j$ with $0\leq j \leq h'-n$;
\item[-] in particular, $\upsilon_J\in \LP_p$ and $A^{\ell'} \upsilon^{(\ell)} \in\LP_p^n$ for each natural $\ell'\leq \ell$ (since $\ell+n< hs+1\leq J\leq h'$);
\item[-] moreover, $\upsilon^{(\ell)}\in Left(\LP_p^n, d^{(\ell)})$;
\item[-] finally, $A^{\ell'} \upsilon^{(\ell)} \in Left^*(\LP_p^n, d^{(\ell)})$ for each natural $\ell'\leq \ell$ 
(since $J$ is a pick);
\end{enumerate}
In this way an integer $d^{(\ell)}$ and an element $\upsilon^{(\ell)}$ with the desired property have been exhibited. 
\end{proof}
\begin{lemma}
\label{transpose}
Let $A\in\LP_p^{n\times n}$ be the matrix such that its transpose $A^T$ is the companion matrix of any monic polynomial 
$-\coeff_0 +  \ldots - \coeff_{n-1} t^{n-1}+ t^n$
from $\LP_p[t]$, where $p$ and $n$ are any two naturals such that $p$ is prime and $n>1$, and let $\glorule$ be the LCA over $(\Z/p\Z)^n$ having $A$ as associated matrix. 
The LCA $\glorule$ is positively expansive if and only if the LCA $\mathcal{G}$ over $(\Z/p\Z)^n$ having $A^T$ as associated matrix is positively expansive. 
\end{lemma}
\begin{proof}
Clearly, $A$ is as in~\eqref{riga}. It holds that $A^TQ=QA$, where
\[Q=
\begin{pmatrix}
-\coeff_{1} &-\coeff_{2}&\cdots&-\coeff_{n-1}&1\smallskip \\ 
-\coeff_{2}&-\coeff_{3}&\iddots&1&0\\ 
\vdots&\iddots&\iddots&\iddots&\vdots \smallskip\\ 
-\coeff_{n-1}&1&\iddots&0&0 \smallskip \\ 
1&0&\cdots&0&0
\end{pmatrix}\in \LP_p^{n\times n}\enspace.
\]
Since $\det(Q)\neq 0$, the thesis directly follows from Lemma~\ref{main2}.
\end{proof}
We now prove that the decidable characterization of positive expansivity provided by Lemma~\ref{exr} also holds also in a more general situation, namely, for LCA over $(\Z/p\Z)^n$ with associated matrix that is in a rational canonical form possibly consisting of more than one block.
\begin{lemma}
\label{riscan}
Let $C\in\LP_p^{n\times n}$ be any matrix in rational canonical form, where $p$ and $n$ are any two naturals such that $p$ is prime and $n>1$, and let $\mathcal{G}$ be the LCA over $(\Z/p\Z)^n$ having $C$ as associated matrix. 
The LCA $\mathcal{G}$ is positively expansive if and only if $C$ is expansive. 
\end{lemma}
\begin{proof}
Let $C_{\pi_1}\in\LP_p^{n_1\times n_1}, \ldots, C_{\pi_s}\in\LP_p^{n_s\times n_s}$ with $n_1+\ldots +n_s=n$ be the diagonal blocks inside $C$, where each $\pi_i$ is the monic polynomial defining $C_{\pi_i}$, \ie, $C_{\pi_i}$ is the companion matrix of $\pi_i$.  Since $\chi_C(t)=\Pi_{i=1}^s \pi_i(t)$ and $\mathcal{G}$ is just the product  $\mathcal{G}_1 \times \ldots \times \mathcal{G}_s$, where each $\mathcal{G}_i$ is the LCA over $(\Z/p\Z)^{n_i}$ having $C_{\pi_i}\in\LP_p^{n_i\times n_i}$ as associated matrix, it follows that $\mathcal{G}$ is positively expansive if and only if every $\mathcal{G}_i$ is positively expansive if and only if, by Lemma~\ref{exr} and~\ref{transpose}, every $C_{\pi_i}$ is expansive, \ie, by Definition~\ref{exppoly}, if and only if every $\pi_i$ is expansive, \ie, by Lemma~\ref{pura}, if and only if $\chi_C$ is expansive, \ie, if and only if $C$ is expansive.
 %
\end{proof}

We are now able to prove Theorem~\ref{risp}.
\begin{proof}[Proof of Theorem~\ref{risp}]
Let $A\in\LP_p^{n\times n}$ be the matrix associated with $\glorule$. By Lemma~\ref{main1}, there exist two matrices $Q, C\in\LP_p^{n\times n}$ such that $\det(Q)\neq 0$, $C$ is in rational canonical form, $QA=CQ$, and $\chi_A=\chi_C$. Let $\mathcal{G}$ be the LCA over $(\Z/p\Z)^n$ having $C$ as associated matrix. By Lemma~\ref{main2}, $\glorule$ is positively expansive if and only if $\mathcal{G}$ is positively expansive, \ie, by Lemma~\ref{riscan}, if and only if $C$ is expansive, \ie, since $\chi_A=\chi_C$, if and only if $A$ is expansive.  
\end{proof}
\end{document}